\documentclass[11pt]{article}

\usepackage[inline]{enumitem}
\usepackage{amsmath}
\usepackage{amsthm}

\usepackage[breaklinks=true,
            colorlinks = true,
            linkcolor = blue,
            urlcolor  = blue,
            citecolor = blue,
            anchorcolor = blue]{hyperref}
\usepackage{breakurl}

\newcommand{\mailto}[1]{\href{mailto:#1}{\nolinkurl{#1}}}

\title{Interactive coin offerings}

\author{
Jason Teutsch\\
\emph{TrueBit Establishment}\\
\mailto{jt@truebit.io}
\and
Vitalik Buterin\\
\emph{Ethereum Foundation}\\
\mailto{vitalik@ethereum.org}
\and
Christopher Brown\\
\emph{Modular, Inc.}\\
\mailto{christopher@modular.network}
}

\date{December 11, 2017\footnote{This version, updated in 2019, includes an Epilogue (Section~\ref{sec:epilogue}).}}

\theoremstyle{plain}
\newtheorem*{prop}{Proposition}
\theoremstyle{definition}
\newtheorem*{defn}{Definition}

\numberwithin{equation}{section}

\begin{document}

\maketitle

\begin{abstract}
Ethereum has emerged as a dynamic platform for exchanging cryptocurrency tokens.  While token crowdsales cannot simultaneously guarantee buyers both certainty of valuation and certainty of participation, we show that if each token buyer specifies a desired purchase quantity at each valuation then everyone can successfully participate.  Our implementation introduces smart contract techniques which recruit outside participants in order to circumvent computational complexity barriers.
\end{abstract}

\section{A crowdsale dilemma} \label{sec:dilemma}

This year has witnessed the remarkable rise of token crowdsales.  Token incentives enable new community structures by employing novel combinations of currency rewards, software use rights, protocol governance, and traditional equity.  Excluding Bitcoin, the total market cap of the token market surged over 60 billion USD in June 2017\footnote{\url{https://coinmarketcap.com/charts/}}.  Most tokens originate on the Ethereum network, and, at times, the network has struggled to keep up with purchase demands.  On several occasions, single crowdsales have consumed the network's entire bandwidth for consecutive hours.

Token distributions can take many forms.  Bitcoin, for example, continues to distribute tokens through a competitive, computational process known as \emph{mining}.  In this exposition, we shall concern ourselves exclusively with the now typical situation in which an anonymous class of buyers wishes to purchase yet-to-be-generated ERC20 tokens over the Ethereum network in exchange for Ethereum's native currency.  Unlike an equity distribution event in which prospective buyers can estimate share values based on existing and potential future revenue streams, tokens sales may not project any revenue at all.  Since traditional analysis fails to estimate initial market valuation for new tokens, buyers must rely on new signals and methods for determining market prices.  The token issuer, on the other hand, faces the unprecedented challenge of not knowing her buyers.  In particular, she cannot tell whether or not two distinct purchasing addresses belong to the same person.

The Ethereum community has experimented with various sale configurations for ERC20 tokens.  In a \emph{capped sale}, for example, the project issuing the ERC20 token announces a fixed price for each new token as well as the maximum (and minimum) number of tokens to be sold.  Capped sales can reach tens of millions of dollars and sell out in a matter of minutes, leaving buyers unable to participate, disappointed, and frustrated.  Uncapped sales, which run without such maximums, provide buyers little clue as to the fraction of total tokens their contribution will ultimately purchase.  Other distribution experiments, including hidden caps and reverse Dutch auctions, have suffered similar fates~\cite{buterinsdilemma}.  Indeed increasing purchase power and limited supply may cause buyers in a reverse Dutch auction to jump in too soon.

Recently, Buterin~\cite{buterinsdilemma} distilled two desireable, mutually exclusive properties of crowdsales.
\begin{prop}
No token crowdsale satisfies that both:
\begin{enumerate}[label=\textup{(\textsc{\roman*})}]
\item a fixed amount of currency buys at least a fixed fraction of the total tokens, and

\item everyone can participate.
\end{enumerate}
\end{prop}
\begin{proof}
If one unit of currency purchases at least $p$~fraction of the tokens, then the total sale revenue cannot exceed $1/p$.
\end{proof}

Clearly any fixed valuation scheme cannot guarantee universal participation, however, we shall construct a crowdsale protocol such that, if each participant specifies a desired purchase quantity at each valuation, then the ultimate token cost to percentage ratio satisfies all buyers (with respect to both valuation and participation).

\section{The bypass} \label{sec:bypass}

Our interactive construction aims to establish an equilibrium of purchase amounts whose sum is satisfactory to all buyers at some uniform valuation.  Given that a liquid market for the new token does not exist prior to the crowdsale, we shall use the crowdsale process itself to reach a common view of the token's present value.  Our protocol will be fair in the sense that parties endowed with large amounts of capital or Ethereum mining power cannot gain either participation advantage or lower cost per percentage of total tokens.  Furthermore, incentives will counteract buyers' resistance to partake in an initially illiquid market. Lastly, we remark that the crowdsale is a game of perfect information in that the public has guaranteed access to all sales data.

Our approach distinguishes itself from prior crowdsale distributions in that:
\begin{enumerate}
\item buyers can \emph{withdraw} their contributions after committing them to the sale (within certain limits), and

\item the protocol exploits sophisticated bookkeeping capabilities of smart contracts.
\end{enumerate}
The corresponding crowdsale is interactive in the sense that potential buyers may enter and exit the crowdsale based on behaviors of other buyers and in doing so tend the valuation towards a market equilibrium.  The protocol also allows sufficient time for informal, social interactions.

The interactive crowdsale begins after the token issuer deploys a smart contract on Ethereum's blockchain.  For the purposes of this discussion, a \emph{smart contract} is a universally trusted machine that, over time, takes in and pays out Ethereum's native currency and deploys a newly generated token.  A combination of messages and native currency from pseudonymous Ethereum \emph{addresses}, each held by some potential buyer, algorithmically determines these payments and deployments.  The smart contract thus effectively collects and retains the crowdsale's token balances.  Through their respective addresses, buyers purchase with native tokens and eventually receive newly minted crowdsale tokens in return.

New messages from addresses to smart contracts are broadcast via \emph{blocks} which occur at regular intervals (approximately every 15--20 seconds\footnote{\url{https://etherscan.io/chart/blocktime}}).  In this way, blocks measure the passage of time on the Ethereum network.  The crowdsale's smart contract can take the current block number as input and therefore alter its behavior as a function of time.

We now describe the components of a simple, interactive coin offering.  We shall present a more detailed specification in Section~\ref{sec:icoprotocol}.

\begin{description}
\item[Basic step:] In each block epoch, buyers can either purchase tokens or voluntarily withdraw funds from the crowdsale.  Buyers specify a maximum sale valuation at which they are willing to participate, and if the sale amount ever reaches this personal threshold, the buyer's \emph{bid} is canceled and she receives a refund.  In Section~\ref{sec:personalmin}, we add support for  bid activation triggered by sale lower bounds.

\item[Withdrawal lock:]  After a certain number of blocks, voluntary withdrawals are no longer permitted.  In a 30-day crowdsale, for example, the smart contract might permit voluntary withdrawals during the first 20 days, but during the last 10 days, only automatic withdrawals are allowed.

\item [Inflation ramp:] Buyers who purchase tokens early receive a discounted price. The maximum bonus might be 20\% (a typical amount for crowdsales today).  The bonus decreases smoothly down to 10\% at the beginning of the withdrawal lock, and then disappears to nothing by the end of the crowdsale.
\end{description}

Individual buyers may submit multiple bids in the crowdsale in order to indicate distinct bid amounts for various valuations.  In particular, they may choose personal thresholds exceeding the total amount of currency in circulation in order to guarantee a successful bid.  The time prior to the withdrawal lock provides an opportunity for buyers to calibrate their purchase amounts, and the period after the withdrawal lock pushes the sale valuation to converge towards an equilibrium value. The inflation ramp reduces entrance inertia and encourages formation of a liquid market.

\section{Cast of characters} \label{sec:coc}

We shall assume certain uniform characteristics among crowdsale participants which will enable us to derive crowdsale invariants and security properties in Section~\ref{sec:analysis} and provide justification for the heuristics discussed in Section~\ref{sec:heuristics}.

\subsection{Buyers} \label{sec:buyers}

Buyers express their individual participation goals by submitting a series of bids, each of which constitutes a ``step'' function.

\begin{defn}
A buyer's \emph{valuation table} is a piecewise step function from the crowdsale's total sale amount to the buyer's contribution amount.
\end{defn}

Buyers may wish to purchase tokens for many different reasons, and we make no specific assumptions in this regard.  For the purposes of our model, all buyers satisfy the following properties.

\begin{enumerate}
\item \emph{Demand is inverse to supply.}  The total sale amount affects individuals' inclinations to contribute, and in a liquid market, buyers will purchase at least as many tokens at lower valuations as they will at higher ones.  In Section~\ref{sec:personalmin}, we shall relax the latter part of this assumption and modify the core procotol of Section~\ref{sec:icoprotocol} accordingly.  We explicitly do not assume that individuals agree on the volume of trade that constitutes a liquid market nor on  the valuation threshold(s) at which contribution inclinations decline.

We shall show in Section~\ref{sec:pci} that each buyer's cumulative purchases ultimately and effectively converge to a monotonically decreasing valuation table under our assumption that buyer demand decreases at higher valuations.  Buyers failing to match this profile would complicate our protocol and analysis.  First, accommodating and monitoring lower entry bounds in the protocol places additional computational stress on the crowdsale smart contract.  In Section~\ref{sec:personalmin}, we discuss an efficient workaround.  Second, buyers can anyway achieve an effect similar to a non-decreasing valuation table by manually purchasing additional tokens later in the sale once the total sale amount has reached the target threshhold.  If sufficiently many buyers were to apply this strange strategy, however, the network might become congested from a positive feedback loop, and then these strategies would fail to execute properly.

\item \emph{Preference for liquid markets.}  Buyers have intrinsic inertia against entering a new crowdsale.  Tokens held by few owners may be difficult to exchange and therefore have uncertain value.  Given the risks of purchasing first, and barring other incentives, most buyers prefer to wait for others to purchase before they do.  Waiting times may vary from buyer to buyer.  We discuss other possible sources of inertia and their circumvention in Section~\ref{sec:personalmin}.

\item \emph{Reliance on social influences.}  Buyers depend on social influences to make purchase decisions.  Since the immediate value of new tokens depends largely on others' beliefs, buyers necessarily interact, either directly or indirectly, with other buyers.  At the beginning of a crowdsale, Buyers lack reliable information with which to valuate the new token.  Social gossip moves at a much slower pace than pure algorithmic trading and consequently dictates the crowdsale pace.  Finally, we assume that social interactions will lead each buyer to eventually, but well before the end of the crowdsale, converge to a final valuation table.

\item \emph{Preference for simplicity.}  Complex procedures for purchasing tokens decreases participation.  The tolerable threshold varies from buyer to buyer, and particular sets of rules or steps may encourage or discourage certain types of buyers.

\item \emph{Pseudonymity.}  Buyers need not disambiguate their identities in order to participate in the crowdsale.  In fact, we expect each buyer to compose her valuation table with bids from multiple pseudonymous addresses.
\end{enumerate}

\subsection{Adversaries} \label{sec:adversary}

We define an \emph{adversary} to be any entity which performs network actions, including purchases and withdrawals, in order to decrease his cost per percentage of total tokens.  We assume that the adversary has significant financial and mining resources, but not enough to create an extended denial-of-service attack which prevents other bids from entering the crowdsale.  In particular, we shall assume that the amount of time that the adversary can sustain significant congestion or censorship on the network is negligible with respect to the duration of the crowdsale.  We also assume that the adversary restricts his actions to the Ethereum network.  He cannot, for example, physically restrain other buyers who wish to participate in the crowdsale.  Finally, we assume that the  crowdsale smart contract always processes bids correctly.

\subsection{Heuristics} \label{sec:heuristics}

The present interactive protocol has two aims.
\begin{enumerate}
\item Maximize useful market information available to buyers at the time of purchase.
\item Provide a fair distribution in line with the conclusion of Section~\ref{sec:dilemma}.
\end{enumerate}
We explicitly do not optimize for maximal valuation.  It remains an open problem to analyze the extent to which psychological and rational forces in the interactive coin offering model impact valuation relative to traditional auction methods.

Interactive coin offerings permit buyers to alter their bids in reaction to signals from other buyers.  As a result, adversaries may intentionally broadcast deceptive or confusing signals in order to achieve pricing gains.  In Section~\ref{sec:blackout}, we give examples of such attacks and explain how the interactive protocol incentivizes against them.  The protocol's pre-withdrawal lock phase aims to maximally reveal useful market information while minimizing buyers' commitments.   While, by design, the protocol allows buyers to change their mind based on new market information, it also must penalize buyers who intentionally mis-signal in order to ensure meaningful market updates.  In order for the protocol to enforce penalties, it must force each buyer to make some formal commitment with each bid.  Intuitively, in order to reach price convergence, buyers must become increasingly confident about market information over time.  Thus, as a prerequisite, buyers' commitments must increase over time.  Therefore the protocol gradually decreases the fraction of active crowdsale funds available for voluntary withdrawals.

The latter phase of the protocol, after the withdrawal lock, allows the protocol to satisfy aim~2 above while avoiding attacks from ``whales'' (see Section~\ref{sec:mvi}).   It does this by establishing two invariant properties, namely that the sale satisfies all buyers in the sense of Section~\ref{sec:dilemma} (and more formally Section~\ref{sec:pci}), and that valuation monotonically increases as shown in Section~\ref{sec:mvi}.

\section{ICO protocol} \label{sec:icoprotocol}

We now describe the operations of the crowdsale's smart contract.  For the purposes of this presentation, \emph{valuation} of the crowdsale refers the total value of tokens sold with respect to the native currency as opposed to the value of the total number of tokens generated.  We shall assume that any tokens created but not sold in connection to the crowdsale event represent a fixed fraction of the total number of tokens generated.  Therefore, regardless of valuation, a given fraction of crowdsale tokens represents a fixed fraction of the total tokens generated.

As new bids enter the crowdsale (Step~1 below), the protocol nullifies active bids with minimal personal caps (Step~3).   Step~3.3 issues partial refunds in order to enforce a monotone valuation invariant  (see Section~\ref{sec:mvi}).  Consequently, in case several bids  tie for minimal personal cap, the protocol refunds an equal fraction of each.  Buyers may stagger their bid values so as to avoid such ties, however, the monotone valuation invariant persists even without such action.

In the protocol below, Step~2 applies only to the phase prior to the withdrawal lock, while Step~3 applies once the withdrawal lock is in effect.

\begin{description}
 \item [Initialization.] All addresses have ``inactive'' status.  Fix time thresholds $0 \leq t < u$, where $t$ corresponds to the ``withdrawal lock'' threshold of Section~\ref{sec:bypass}, and set $s=0$.  Let $p(s)$ be a positive-valued, linear\footnote{Section~\ref{sec:blackout} makes use of linearity in conjunction with Step~1.2 and Step~2.2 below.}, decreasing function representing the purchase power of a native token at stage~$s$.  Finally, define the \emph{crowdsale valuation at the present instant} as follows.
\[
V =
\begin{cases}
0 & \text{if no addresses are active;}\\
\sum_{\text{$A$ active}} v(A) & \text{otherwise.}
\end{cases}
\]
Here $A$ is an address, and $v(A)$ is a function mapping addresses to quantity of tokens as specified below.

\item[Main Loop.] The following four steps are repeated in each block while $s \leq u$, where $u$ delimits the end of the crowdsale.
\begin{itemize}
\item[\textsc{Step~1}:] \textsc{Receive bids.}
\begin{enumerate}
\item Any ``inactive'' address $A$ may send to the crowdfund smart contract:
\begin{itemize}
\item  a positive quantity of native tokens $v(A)$ along with
\item  a positive-valued \emph{personal cap} $c(A) > 0$.  In case $s \geq t$, i.e.\ when the withdrawal lock is in effect, we require the stricter inequality $c(A) > V$.
\end{itemize}

\item The smart contract then
\begin{itemize}
\item sets the address balance $b(A) = v(A) \cdot p(s)$, effectively implementing the inflation ramp (Section~\ref{sec:bypass}), and
\item sets $A$'s status to ``active.''
\end{itemize}
\end{enumerate}

\item[\textsc{Step~2}:] \textsc{Voluntary withdrawals (execute this step iff $s < t$)}.

The following only applies prior to the withdrawal lock at time~$t$.  Any ``active'' address $A$ may signal that it wishes to cancel its bid from any previous stage.   Upon such signal, the crowdfund smart contract does the following:
\begin{enumerate}
\item refunds $v(A) \cdot (t-s)/t$ native tokens back to $A$,
\item sets
\[
b(A) = v(A) \cdot s/t \cdot \left[p(a) - \frac{p(a) - p(u)}{3}\right],
\]
where $a$ denotes the time at which address $A$ originally made its bid,
\item sets $A$'s status to ``permanent.''
\end{enumerate}
The three steps above refund a fraction of $A$'s capital to the buyer and permanently commit a fraction of $A$'s capital to the sale while scratching a third of the ``bonus'' pricing.  In Section~\ref{sec:blackout}, we shall argue that our chosen parameter of one third in Step~2.2 suffices to deter rational attacks.  Note that the fraction of capital committed after a voluntary withdrawal increases as the withdrawal lock approaches.  

\item [\textsc{Step~3}:] \textsc{Automatic withdrawals (execute this step iff $s \geq t$)}.

While there exists an active address~$B$ whose personal cap is exceeded by the present crowdsale valuation, i.e. $V > c(B)$,
repeat the following three steps.  We update the value~$V$ only after the dashed bullets in each iteration of the while loop, regardless of whether or not \eqref{eqn:meatball} holds.
\begin{enumerate}
\item Let $B_1, \dotsc, B_k$ be the (distinct) active addresses with minimal personal cap at the present moment, i.e.
\[
c(B_i) = \min \{c(A) : \text{$A$ is active or permanent}\}
\]
for all $i \leq k$, and let
\[
S = \sum_{i=1}^k v(B_i).
\]
\item If removing the bids of $B_1, \dotsc, B_k$ does not suffice to satisfy all personal caps from active addresses, i.e. 
\begin{equation} \label{eqn:meatball}
V - S \geq c(B_1),
\end{equation}
then the crowdsale smart contract kicks out the entirety of these bids.  In more detail, the smart contract:
\begin{itemize}
\item refunds $v(B_i)$ to $ B_i$ for all $i \leq k$, and 
\item sets the statuses of $B_1, \dotsc, B_k$ to ``used.''
\end{itemize}
\item Otherwise, the reverse inequality of \eqref{eqn:meatball} holds, and only some fraction of each of $v(B_1), \dotsc, v(B_k)$ comes out of the crowdsale.  Let $0 <q < 1$ be the minimum (positive) fraction of these quantities that must be removed in order to satisfy all remaining personal caps, i.e.
\[
q =
\frac{V - c(B_1)}{S}.
\]
For all $i \leq k$, the crowdfund smart contract:
\begin{itemize}
\item refunds $q \cdot v(B_i)$ to $B_i$,
\item sets the new value of $v(B_i)$ to be $(1 - q) \cdot v(B_i)$, and
\item sets the new value of $b(B_i)$ to be $(1- q) \cdot b(B_i)$.
\end{itemize}
The $B_i$'s remain ``active,'' and the total crowdsale valuation is now exactly $c(B_1)$ because the procedure above removes exactly $qS$ native tokens from the crowdsale's smart contract.
\end{enumerate}
\item[\textsc{Step~4}:] Increment~$s$.
\end{itemize}
\item[Final Stage.] Each ``active'' or ``permanent'' address~$A$ receives $b(A)$ tokens at the end of the crowdsale.
\end{description}

Note that the loop in Step~3 eventually terminates because $V$ decreases with each iteration while $\min \{c(B) : \text{$B$ is an active address}\}$ increases.

\section{Analysis} \label{sec:analysis}

We conclude by highlighting some properties of the interactive coin offering protocol detailed in Section~\ref{sec:icoprotocol}.  Our analysis relies on two, key, quantitative invariants which come into effect after the withdrawal lock (Section~\ref{sec:bypass}), namely that valuation is monotonically increasing over time, and that all personal caps remain above the current valuation in each block.

\subsection{Personal cap invariant} \label{sec:pci}

We argue that the final crowdsale valuation and purchase amounts, after all automatic withdrawals from the pre-withdrawal phase have completed, satisfy every buyer's valuation table (see Section~\ref{sec:buyers}).   In order to parse this statement, we need to explain two things.
\begin{enumerate}
\item How does the buyer's cumulative purchases from various addresses formally correspond to a valuations table?

\item What does it mean for a crowdsale to ``satisfy'' a valuation table?
\end{enumerate}

Let $V$ denote the final valuation of the crowdsale.  Regarding item 1., note that the net purchase effect $\mathcal{A}$ of a bid from an address~$A$ is a single-step valuation table (modulo a single point):
\[
\mathcal{A}(V) =
\begin{cases}
v(A) & \text{if $V < c(A)$;} \\
\text{some value in $[0, v(A)$]} & \text{if $V = c(A)$;}\\
0 & \text{if $V > c(A)$.}
\end{cases}
\]  
We cannot specify a definite value for $\mathcal{A}[c(V)]$ because the purchase amount at this valuation depends on the bids made by other addresses.  $A$ could either receive a partial refund in Step~3 of Section~\ref{sec:icoprotocol}, or none at all.  The buyer's cumulative purchase is determined by the sum of his bids.  In other words, the buyer's (stepwise) valuation table equals the sum of the single-step valuation tables for his addresses, which answers item~1.  As noted in Section~\ref{sec:buyers}, the valuation table graph restricted to the right of the last bid made by the buyer is a monotonically decreasing function because single-step bid functions may end but not begin after that point.

We now turn our attention to item~2.  We say that a valuation $V$ and purchase amount $a$ satisfy a buyer's valuation table $T$ if the following holds:
\begin{enumerate} [label=({\alph*})]
\item  $a = T(V)$ if $V$ is an interior point of some valuation table step;

\item otherwise $a$ lies somewhere between the left and right limit points:
$$\lim_{x \to V^+} T(x) \leq a \leq  \lim_{x \to V^-} T(x).$$
\end{enumerate}
There are two cases to consider.  If the final valuation $V$ from the crowdsale does not equal $c(A)$ for any of the buyer's active addresses $A$, then $V$ is not a limit point of any of the single-step valuation tables for the buyer's addresses, and $T(V)$ is simply the sum of the single-step valuation tables evaluated at $V$.  In this case, $T(V)$ matches the purchase amount exactly, satisfying~(a) above.  Otherwise, $V$ matches the personal cap for some address(es).  Then $V$ is a limit point of some step on the buyer's valuation table, and the actual purchase amount is the sum of $v(A)$ over active addresses $A$ whose personal cap does not equal $V$ plus some fraction of the $v(A)$'s for addresses $A$ whose personal cap matches $V$, yielding~(b) as desired.

In summary, the crowdsale satisfies the buyer's valuation table which closely resembles the sum of her bids.

\subsection{Monotone valuation invariant} \label{sec:mvi}
We shall demonstrate that regardless of what bids buyers submit to the crowdsale smart contract, valuation is monotonically increasing after the distinguished time threshold~$t$ (see Section~\ref{sec:icoprotocol}).  After time~$t$, voluntary withdrawals do not occur, so we need only consider the other protocol steps.  The loop condition in Step~3 guarantees that the valuation at the beginning of Step~1 is less than or equal to the personal cap for each active address.  Indeed the active addresses in each iteration of the loop in Step~3 are a subset of those addresses which were active at the end of Step~1, and furthermore the valuation at the end of Step~3 is no less than the personal cap of every active address (regardless of whether the loop terminated after Step~3.2 or Step~3.3).  To recap, the valuation at the beginning of Step~1 is less than or equal to the minimum of all active addresses at that time, which is less than or equal to the valuation at the end of Step~3, which proves the claim.

The invariant property above allows the interactive coin offering to resist ``pushout attacks'' from rich buyers, or \emph{whales}, of the following form.  Suppose there are two existing bids purchasing 30 tokens each, and each bid has a personal cap of 79 tokens.  Now say that a whale bids a 50 token purchase with a large personal cap of 200 tokens.   The total of all bid amounts would now be 110 tokens, exceeding the personal caps of the original two bids.  If those two bids were to come out, the whale would have a bid with a valuation of 50, which is lower than the original valuation of 60 tokens.  The invariant above proves this cannot happen.

\subsection{Interactive blackouts} \label{sec:blackout}

Ulrich Gall\footnote{\href{https://medium.com/@ulrichgall/please-help-me-understand-before-the-withdrawal-lock-bidders-can-just-completely-withdraw-their-dd95792f8e2e}{{\tt https://medium.com/@ulrichgall/please-help-me-understand-before-the-}\\{\tt withdrawal-lock-bidders-can-just-completely-withdraw-their-dd95792f8e2e}}} pointed out that an inflation ramp (see Section~\ref{sec:bypass}) is not incentive compatible with allowing buyers to freely withdraw their bids.  Consider what would happen if an adversary were to place an enormous bid at the beginning of the crowdsale with the intention of voluntarily withdrawing most of its capital in the moments preceeding the withdrawal lock.  On account of the inflation ramp, the adversary would receive a \emph{bonus} for his early participation, while other truthful buyers, who wait to place their bids in response to a deceptively high valuation, would not obtain bonuses.  Consequently, this attack strategy affords the adversary a price advantage relative to other buyers.  Since the adversary stands to gain from broadcasting misinformation, the crowdsale protocol must incentivize against his action.

The observation above directly justifies the lower bound for personal caps given in Step~1.1 of the crowdsale protocol (Section~\ref{sec:icoprotocol}).  While the requirement $c(A) > V$ is used in the proof of the monotone valuation invariant (Section~\ref{sec:mvi}), imposing such a condition during the voluntary withdrawal phase might permit an adversary to suppress market signaling via the strategy described in the previous paragraph.  We therefore relax this constraint to $c(A) > 0$ prior to the withdrawal lock and forbid automatic withdrawals during this period.

We now attempt to quantify the adversary's advantage from executing the attack above.  The amount of gain enjoyed by the attacker depends on the other buyers' strategies.  We shall make a conservative assumption that only a single adversary uses the attack strategy and that the remaining buyers truthfully respond to market signals in the sense that each submits a bid if and only if she finds the current valuation lower than her ``true'' valuation at that given instant.  Let $a$ denote the adversary's bonus for early participation, and let $b$ be the bonus which other truthful buyers obtain from their purchases moments after the withdrawal lock.  In the example parameters given in Section~\ref{sec:bypass}, $a$ is at most $20\%$ and $b$ is $10\%$.  Let $x$ denote the adversary's initial capital contribution at bonus~$a$, and let $y$ denote the capital contributions of all other buyers at bonus~$b$.

If the adversary had not employed the \emph{blackout attack} above, then all buyers, including the adversary himself, might obtain bonus $a$, resulting in
\begin{equation} \label{eq:truthful}
\frac{(1+a) \cdot x}{(1+a) \cdot(x+y)}
\end{equation}
fraction of the tokens going to the ``adversary.''  If other buyers decided to buy in later than the first block of the crowdsale, then the attacker might receive even more.  By executing the blackout attack above, on the other hand, the adversary collects a larger fraction of the tokens, namely
\begin{equation} \label{eq:blackout}
\frac{(1+a) \cdot x}{(1+a) \cdot x + (1+b) \cdot y}.
\end{equation} \label{eq:difference}
For clarity, we simplify notation as follows.  Let $A = 1+a$ and $B = 1+b$.  Now the adversary's advantage is the difference between the percentages \eqref{eq:blackout} and \eqref{eq:truthful}, namely
\begin{equation} \label{eq:difference}
\frac{Ax}{Ax+By} - \frac{Ax}{Ax+Ay} = Axy \cdot \frac{A-B}{A^2x^2 + A^2xy + ABxy + ABy^2}.
\end{equation}
There are two cases to consider.  If $x \leq y$, then we can use the fact that $ABxy \leq ABy^2$ and $0 \leq A^2x^2$ to obtain the following upper bound for the right-hand side (RHS) of \eqref{eq:difference}:
\[
\frac{A - B}{A + 2B}.
\]
If $x \geq y$, on the other hand, then, by a similar method, then RHS is bounded above by
\[
\frac{A - B}{2A + B}.
\]
In either case, the adversary's gain is at most $(a-b)/3$ fraction of tokens, hence a withdrawal penalty of this amount suffices to deter him if he is rational.  We remark that if the adversary were to commence his attack later during the voluntary withdrawal period, his advantage would be even less.

Finally, we analyze the penalty in the protocol which discourages this attack.  According to Step~2 of the crowdsale protocol (see Section~\ref{sec:icoprotocol}) a buyer who executes a voluntary withdrawal must also permanently commit a fraction $p$ of his capital contribution to the crowdsale, where $p$ equals the fraction of the voluntary withdrawal epoch which has elapsed at the moment of withdrawal.  Assume, for the moment, that the adversary initiates his blackout attack via a purchase in the first block of the crowdsale, and let~$a$ be the bonus amount at this time (e.g.\ 20\%).  Let $b$ be the bonus available at the instant when truthful buyers would theoretically observe and react to the adversary's withdrawal, (e.g.\ 10\% at the instant of the withdrawal lock).  The adversary's withdrawal penalty is then at least $ap/3$, since he forfeits $p/3$ fraction of his bonus in accordance with Step~2.2 of the protocol (Section~\ref{sec:icoprotocol}).   Step~2 actually has a considerably stronger penalty effect since the adversary can no longer recover his initial bonus.  As argued in the previous paragraph, the adversary loses money whenever this penalty exceeds the threshold $(a - b)/3$.  Combining this fact with our observation about $ap/3$, we see this loss occurs whenever~$p$ satisfies
\begin{equation} \label{eq:ab3a}
p > \frac{a-b}{a}.
\end{equation}
Let $p_1$ denote the fraction on the right-hand of \eqref{eq:ab3a}.

We analyze the time bound further.  Since other truthful buyers know that a rational adversary would have withdrawn by time~$p_1$, they can reliably bid after this point based on current market information.  It follows that a rational, truthful buyer would bid at time~$p_1$ rather than waiting until the bonus drops all the way to $b$, and so it makes sense to recursively reapply the argument using the bonus for truthful victims at time $p_1$.  Assuming a linearly decreasing pricing scheme, the bonus at time $p_1$ is $b_1 = a - p_1(a-b)$ by Step~1.2 of the protocol.  Therefore the adversary loses money whenever his withdrawal penalty is greater than $(a-b_1)/3$, whence, by the same argument used to derive \eqref{eq:ab3a}, the rational adversary surely would have withdrawn before time
\[
p_2 = \frac{a-b_1}{a} = \frac{p_1(a-b)}{a}.
\]
Iterating, we obtain
\[
p_n = \left(\frac{a-b}{a}\right)^n.
\]
This quantity vanishes as $n$ tends to infinity because the geometric term is strictly bounded above by~1.  Therefore, assuming a recursive application of rationality as illustrated above, the blackout attack will never succeed.  

Without loss of generality, we may assume that the adversary broadcasts his bid after the first block of the crowdsale.  The argument proceeds as before, except we measure time relative to that later block instead.  This assumption does no harm because the ``true'' withdrawal penalty, or more precisely capital commitment, dominates the simulated penalty one would obtain from measuring time relative to the moment of actual withdrawal.

We conclude by observing that the instantaneous market information available prior to the withdrawal lock does not consist of a unique valuation but rather a list of active bids together with some quantity of permanently committed capital.  The form of this information seems roughly consistent with human psychology in the sense that buyers may not, at first, have a specific valuation in mind.  The protocol offers each buyer the opportunity to reach a personal interpretation of available market data through complementary, qualitative channels.

\subsection{Fairness}

The protocol treats large purchases and small purchases uniformly.  Whales with low personal caps get pushed out of the crowdsale in just the same way as buyers who purchase a fraction of a token.  Furthermore, any buyer can freely specify any non-trivial personal cap that they please, and the fees for submitting transactions to the crowdsale smart contract are flat fees based on Ethereum gas prices (See Section~\ref{sec:implementation}).  Finally, the smart contract handles all purchases publicly, which means that all prospective buyers have perfect information about all other bids.

\subsection{Censorship}

In the past, whales have benefited from network congestion during capped crowdsales.  One BAT crowdsale buyer, for example, paid \$6660 towards a single transaction fee to ensure that his transaction entered the current block, effectively preventing others buyers from participating \cite{buterinsdilemma}.  In theory miners can potentially mimic or amplify this bias by censoring transactions during the crowdsale.  While these types of denial-of-service attacks may succeed in quick crowdsales, they become impractically costly over extended crowdsales, such as the one presented here.

\subsection{Last-minute withdrawals}

A whale who bid a huge number of tokens but then withdrew his purchase in the last block of the crowdsale could deter other buyers from participating and thereby obtain an artificially low valuation.  For this reason, the protocol forbids voluntary withdrawals in the latter phase of the crowdsale.  During the first phase of the crowdsale, the fraction of each bid available for voluntary withdrawal decreases smoothly over time, thereby mitigating against the edge case discussed in Section~\ref{sec:blackout} in which one buyer makes a last-minute withdrawal and the remaining buyers do not have an opportunity to react to it.  By incentivizing against last-minute actions, the protocol increases the chance of converging to stable bids and valuation prior to the end of the crowdsale.

\subsection{Overcoming inertia}

The protocol design encourages early participation and maintains a low barrier to crowdsale entry.  An inflation ramp (Section~\ref{sec:bypass}) incentivizes buyers to enter the crowdsale early and form a liquid market.  Moreover, the beginning of the crowdsale offers a low-risk trial period in which buyers can voluntary withdraw their bids with little penalty.  Finally, the crowdsale has a relatively simple user interface.  Anyone who wishes to purchase tokens without risking automatic withdrawals can simply submit a transaction to the crowdsale's smart contract with a personal cap exceeding the total number of native tokens in circulation (a predictable value in Ethereum).

\section{Lightweight implementation} \label{sec:implementation}

In Section~\ref{sec:adversary}, we made the simplifying assumption that ``crowdsale smart contracts always process bids correctly.''  We now refine this simplistic point-of-view.  Computational tasks of minimal complexity, that is, those quantitatively resource bounded by per block \emph{gas limit}\footnote{\url{https://ethstats.net/}}, execute correctly so long as enough \emph{gas}, or block founder payment, accompanies the transaction which initiates the task.  Smart contracts themselves may call tasks so long as: 
\begin{enumerate}
\item the task itself runs within the per block gas limit (and available network bandwidth),

\item the smart contract has sufficient \emph{ether}, or native currency, to pay for the task execution, and

\item the smart contract remains dormant between the blocks in which users interact with it.
\end{enumerate}

The ICO protocol's main loop (Section~\ref{sec:icoprotocol}) requires maintenance of a list of addresses with various personal caps, a way of finding the set of addresses with minimal personal caps, and a mechanism for deleting this set from the list.  Traditional heap data structures require $O(\log n)$ time to execute one or more of these operations on a list of size~$n$.  Hence if the address list grows sufficiently large, then the smart contract cannot maintain the heap without violating item~1 above.  Moreover, per item~2, who pays for each of these operations and when?  An autonomous crowdsale smart contract could become increasingly expensive over time as it acquires additional bids.  It will be clear from inspection that our construction below satisfies item~3.

Given the infeasibility of maintaining a heap data structure within the crowdsale smart contract, one might wonder whether the smart contract could recruit outside parties to perform the necessary heap maintenance through incentives.  Outside parties, however, might supply incorrect data.  How would the smart contract know whether the supplied address's personal cap is minimal?  The smart contract's inability to verify and delete the actual minimum opens a potential attack vector.  Suppose there exists a bid of 50 with personal cap 60, a bid of 35 with personal cap 80, and that the current valuation is 115.  If an attacker reports that 80 is the minimum, active personal cap, then the $(35, 80)$ bid would come out entirely while part of the $(50,60)$ bid remains active.  According to the protocol specification, however, the entire $(50,60)$ should come out while instead the $(35,80)$ bid remains completely active.

Instead of a heap, the crowdsale contract sorts bids via a \emph{linked list}, a data structure wherein each element in the list contains a ``link'' pointing to its immediate successor.  While insertion into a sorted linked list structure with~$n$ elements requires an expensive~$O(n)$ operation, checking an insertion, by inspecting backward and forward neighbors, requires negligible, constant time (in accordance with item~1 above).  Hence the crowdsale smart contract can add a new bid into the linked list using only a small amount of gas plus some verifiably correct advice supplied by an incentivized third-party.  We shall assume that the third-party's cost to execute the advice operation is negligible.  Moreover, the smart contract can itself identify and/or delete the minimal element in an ascending linked list in negligible constant time.

The crowdsale smart contract efficiently maintains a \emph{valuation pointer} which points to the current valuation and which, by the monotone valuation invariant  (Section~\ref{sec:mvi}), moves monotonically forward through the linked list over time.  A bid is active if and only if sits at or ahead of the valuation pointer.  While buyers may submit bids to crowdsale smart contract prior to the withdrawal lock, the protocol cannot meaningfully point to a valuation figure until after the voluntary withdrawal period.   Therefore the crowdsale valuation pointer comes into existence only after the withdrawal lock.

The crowdsale smart contract must accurately update the valuation pointer in each block in order to enforce the inequality $c(A) > V$ from Step~1.1 of the protocol (Section~\ref{sec:icoprotocol}).  Without this property, the crowdsale could end up in the awkward situation where a fresh bid enters with personal cap~$V$ and the crowdsale (should have) already kicked back some fraction of the capital from other bids with personal cap~$V$.

For purposes of gas efficiency, we attempt to minimize the number of moves that the valuation pointer must make in each block.  The astute reader may note that the sums of parameter~$k$ over each iteration in Step~3 of the protocol (Section~\ref{sec:icoprotocol}), i.e.\ the number of bids with minimal personal caps, could be quite large, particularly in blocks where a large amount of capital enters the crowdsale.  Since the gas limit restricts the number of such bids that the protocol can process in a single block (item~1), insufficient gas could result in an incorrect valuation pointer at the next time step.  Indeed, the valuation pointer always moves forward to the bucket containing the current valuation.  In order to avoid lags in valuation pointer updates, we require the bids to sit in evenly-spaced, discrete \emph{buckets}, each containing bids with identical personal caps.  These buckets form the elements of the linked list, and each bucket points to the bucket representing the next largest personal cap.   The protocol saves gas by monitoring the sum of all capital and remaining active fraction in each bucket rather than keeping track of individual, member bids.  These quantities suffice to perform the calculations in Step~3.2 and Step~3.3 of the protocol, and indeed, all bids in a given interval are subject to the same automatic withdrawals.  In summary, when a buyer submits a new bid, she includes sufficient incentive for a third-party to both supply the memory location of the appropriate bucket for her bid and generate a new bucket for it on the fly if needed.

We remark that the bucket approach avoids the contentious situation where two bids are inserted into the same spot in the linked list at the same time, and only the first one makes it into the crowdsale.  We also observe that the granularity of each interval should be greater than the maximum cumulative capital contribution per block (e.g.\ the total native tokens in circulation) divided by the total number of pointer moves possible per block.  The September~2017 gas limit of 6.7~million might support up to 300,000 pointer moves per block.  We estimate this figure based on a fixed cost of 40,000 gas to initiate the pointer moving loop, plus 19 gas per iteration of the loop, plus some gas for submitting the bids.

At the end of the sale, each buyer manually pulls either native tokens or newly generated tokens out of the crowdsale smart contract.  The pulling operation provides greater security over automated pushing since careless coding of the latter form could result in an unintended, incentivized, reentrancy exploit.

\section{Personal minimums} \label{sec:personalmin}

Funding inherently makes a project more valuable.  Starting costs vary but may include personnel, legal fees, marketing, office space, licenses, travel, equipment, or administration.  Buyers may perceive substantially higher risk with tokens associated with an underfunded project.  Furthermore, they may not agree on the fixed cost required to get a venture off the ground.  A crowdsale which permits each buyer to specify a \emph{personal minimum} above which she would wish to participate therefore reduces this entrance risk.  Below we describe a practical mechanism for realizing personal minimums.

Suppose that a crowdsale smart contract receives 3 bids, each with capital contribution 10 and personal minimum 30.  While the crowdsale valuation might never reach the activation threshold of these three bids, each bid would gladly activate were the other two bids to activate first.  How could a crowdsale smart contract recognize this relationship amongst a deluge of other submitted bids?

We describe an incentivized, method for monitoring personal minimums which functions with limited gas resources.  In short, the crowdsale smart contract rewards users who submits a \emph{target valuation}~$x$ and a \emph{target set} of bids such that:
\begin{enumerate}
\item $x$ exceeds all personal minimums of bids in the target set, and
\item the sum of capital contributions from bids in the target set exceeds~$x$.
\end{enumerate}
While the target set here may include all active bids, a nontrivial target set must include inactive bids as well.  The two properties above suffice to justify the activation of all (inactive) bids in the target set, and the smart contract can easily verify these conditions.  We append this operation to Step~1 in Section~\ref{sec:icoprotocol}.

Each bid with a personal minimum includes a flat fee to pay for outside users to \emph{poke} it in via the operation above.  A simultaneous poke of five bids, then, would collect the sum of fees from each of these five bids.  The rewards should suffice to compensate for the fact that identical pokes specifying the same pair $(\text{target valuation}, \text{target set})$ may occur and that only the first receives the reward.  Note that the computational complexity of identifying a target valuation and corresponding set grows as the set of bids grows, so the poking reward should scale over time in a crowdsale with a truly enormous number of bids.  Finally, we remark that an SSTORE call in Ethereum costs 5000~gas, and therefore, under the gas limit of 6.7~million as of September 2017, one could poke in as many as 1300~bids with a single transaction.

We now describe an implementation structure which increases gas efficiency, maintains the protocol's indifference to the order of bid submissions, and allows the protocol to remain resistant against attacks from whales (Section~\ref{sec:mvi}).  The crowdsale protocol places bids into discrete, evenly-spaced buckets based on personal minimums, just as it does for personal caps.  As before, the smart contract offers incentives for advice on where to insert new bids and buckets into the linked list.  Bids transition to \emph{fully active} status when successfully poked, and only fully active bids count towards the crowdsale valuation.  For gas efficiency, the procotol pokes bids by the bucket.  If a bid with a given personal minimum can become fully active, then so can all other bids sharing the same personal minimum.

Finally, as discussed in Section~\ref{sec:implementation}, the crowdsale must enforce the condition $c(A) > V$  on all newly active bids, including those which enter from poking.  We assume that the poking incentive guarantees that all bids become fully active as soon as they become eligible for such status, and therefore that the order of pokes does not materially affect the final set of fully active bids in the crowdsale.  Each bid becomes fully active before the valuation pointer reaches its personal cap bucket and maintains this status until, if ever, the crowdsale valuation reaches the bid's personal cap.

For full implementation details, please see our codebase: 
\begin{quote}
\hspace{-4.5ex}
\small
\url{https://github.com/TrueBitFoundation/interactive-coin-offerings/}.
\end{quote}

\section{Conclusion}

Crowdsales pose critical, timely, and challenging game-theoretic questions.  While reasonable assumptions and deductive reasoning can guide our intuition, ultimately we must also rely on empirical evidence to arrive at definitive cryptoeconomic conclusions and ideal parameter selections.  The protocol discussed herein offers a means for achieving fair valuation equilibrium.

\appendix
\section{Epilogue} \label{sec:epilogue}

Since the initial publication of this work, Truebit developed, together with BKDF, a graphical interface for the interactive coin offering.  The following link provides screenshots and a demo.
\begin{center}
\small
\url{https://medium.com/truebit/exploring-the-iico-interactive-dapp-337e1d09fffe}
\end{center}
This version improves gas efficiency via a ``poke-out'' mechanism and provides the ``poke-in'' functionality outlined in Section~\ref{sec:personalmin}.  The article below details these key changes and contrasts this third-generation system with previous implementations by Modular and Kleros.
\begin{center}
\small
\url{https://medium.com/truebit/an-intro-to-truebits-interactive-coin-offering-e6d1dae36090}
\end{center}

\paragraph{Acknowledgments.} The authors are grateful to Christian Reitwie{\ss}ner for help with implementation details.  We also thank Ryan Zurrer for useful comments.

\bibliographystyle{plain}
\bibliography{ico}

\end{document}